\documentclass[12pt]{article}%
\usepackage[letterpaper, top=1in, bottom=1in, left=1in, right=1in]{geometry}
\usepackage{amsfonts, amsmath, amsthm, amssymb}
\usepackage{enumitem, graphicx, stmaryrd, xcolor, bbm}
\SetSymbolFont{stmry}{bold}{U}{stmry}{m}{n}
\usepackage{url, algorithm}
\usepackage{float}
\floatstyle{ruled}
\usepackage{dsfont}

\usepackage[pagebackref=true,colorlinks]{hyperref}
	\hypersetup{linkcolor=[rgb]{.7,0,.7}}
	\hypersetup{citecolor=[rgb]{.5,0,.5}}
	\hypersetup{urlcolor=[rgb]{.7,0,.7}}

\newtheorem{theorem}{Theorem}[section] 
\newtheorem{lemma}[theorem]{Lemma}

\newtheorem{definition}[theorem]{Definition}

\newtheorem{conjecture}[theorem]{Conjecture}
\newtheorem{corollary}[theorem]{Corollary}
\newtheorem{remark}[theorem]{Remark}						
\newtheorem{proposition}[theorem]{Proposition}

\newcommand{\abs}[1]{\left|#1\right|}		
\newcommand{\E}{\mathop{\mathbb{E}}}  		
\newcommand{\R}{\mathbb{R}}  		
\newcommand{\N}{\mathbb{N}} 			  		

\newcommand{\set}[1]{\left\{ #1 \right\}}   
\newcommand{\brac}[1]{\left( #1 \right)}    
\newcommand{\sqbrac}[1]{\left[ #1 \right]}  
\newcommand{\ind}{\mathds{1}}				

\newcommand{\Inf}{\textnormal{Inf}} 				
\newcommand{\Var}{\textnormal{Var}}				
\newcommand{\up}[1]{^{(#1)}}

\newcommand{\poly}{\mathrm{poly}}
\renewcommand{\AA}{\mathsf{AA}}
\newcommand{\pbias}{\boldsymbol{\mu}_p}
\newcommand{\qbias}{\boldsymbol{\mu}_q}
\newcommand{\halfbias}{\boldsymbol{\mu}_{1/2}}
\newcommand{\pslice}{\boldsymbol{\nu}_p}
\newcommand{\qslice}{\boldsymbol{\nu}_q}
\newcommand{\blambda}{\boldsymbol{\lambda}}

\let\oldparagraph\paragraph
\renewcommand{\paragraph}[1]{\oldparagraph{#1.}}

\title{Distributional Variants of the Aaronson--Ambainis Conjecture}
\author{Uma Girish\thanks{Department of Computer Science, University of Toronto. Email: {\tt uma.girish@utoronto.ca}} \and Kunal Mittal\thanks{Department of Computer Science, Courant Institute of Mathematical Sciences, New York University. Email: {\tt kunal.mittal@nyu.edu}. Research supported a Simons Investigator Award.} \and Barak Nehoran\thanks{Department of Computer Science, Columbia University. Supported by AFOSR award FA9550-23-1-0363, NSF awards CCF-2530159, CCF-2144219, and CCF-2329939, and the Sloan Foundation.} \and Ran Raz\thanks{Department of Computer Science, Princeton University. Email: \texttt{ranr@princeton.edu}. Research supported by a Simons Investigator Award. }}
\date{}			 					
\begin{document}

\maketitle

\begin{abstract}
A longstanding conjecture in quantum complexity theory asserts that, under the uniform input distribution, quantum query algorithms can be polynomially simulated by classical query algorithms. More precisely, the acceptance probability of any quantum query algorithm can be approximated, on average over uniformly random inputs, by a classical query algorithm, with only polynomial query overhead. The conjecture is central to understanding whether exponential quantum advantages for decision problems necessarily rely on additional structure.

We study analogues of this conjecture under other natural input distributions and prove that they are all equivalent to the original uniform-distribution conjecture. We first consider the product distribution $\mu_p$, where the input bits are independent Bernoulli variables with fixed bias $p$. We show that for every fixed $p \in (0, 1)$, quantum query algorithms under $\mu_p$ admit polynomial-overhead classical simulations if and only if the same holds under the uniform distribution. Second, we consider the distribution $\nu_p$ that is uniform over the slice of strings with Hamming weight $\lfloor pn \rfloor$. We show that for every fixed $p \in (0, 1)$, quantum query algorithms under $\nu_p$ admit polynomial-overhead classical simulations if and only if the same holds under the uniform distribution.

The Aaronson--Ambainis conjecture is a stronger statement that implies the  above-mentioned conjecture and is formulated in terms of bounded low-degree polynomials on the Boolean hypercube.
It asserts that under the uniform distribution, any such polynomial with nonnegligible variance must have an influential variable. We formulate analogues of this conjecture, where the underlying distribution is a biased product distribution or a uniform distribution over a slice, and prove that all these variants are equivalent to the original Aaronson–Ambainis conjecture.
\end{abstract}






\section{Introduction}

A fundamental challenge in quantum complexity theory is to understand how much structure is necessary for an exponential quantum advantage. 
Do exponential quantum speedups arise only for highly structured problems, such as Simon's problem~\cite{Sim97} and other period-finding problems,
or can they also arise for decision problems under seemingly unstructured input distributions? More specifically, for which input distributions can quantum query algorithms achieve an exponential advantage over classical query algorithms?

\paragraph{The Folklore Conjecture}
A longstanding conjecture, described by Aaronson and Ambainis~\cite{AA14} as folklore, asserts that, under the uniform input distribution, the acceptance probability of every quantum query algorithm can be approximated on average by a classical query algorithm with only polynomial query overhead.
More precisely, for every $T$-query quantum algorithm and every $\epsilon > 0$, there should exist a classical query algorithm making at most $\poly\brac{T, \frac1\epsilon}$ queries, whose acceptance probability differs from that of the quantum algorithm by at most $\epsilon$, on average, over uniformly random inputs. 
If true, the conjecture would rule out super-polynomial quantum query advantages for decision problems under the uniform distribution.\footnote{
    The picture is very different for search problems, where Yamakawa and Zhandry~\cite{YZ24} showed that there is a problem exhibiting exponential quantum advantage under a uniformly random oracle.
}
The conjecture is therefore central to understanding whether super-polynomial quantum speedups for decision problems necessarily rely on additional structure in the input distribution.

\paragraph{The Aaronson--Ambainis Conjecture}
The acceptance probability of a $T$-query quantum algorithm can be expressed as a bounded multilinear polynomial of degree at most $2T$ on the Boolean hypercube~\cite{BBCMdW01}.
Motivated by this representation, Aaronson and Ambainis 
stated a conjecture about bounded low-degree polynomials, known as the Aaronson--Ambainis conjecture~\cite{AA14}. Informally, their conjecture asserts that every bounded low-degree polynomial with nonnegligible variance under the uniform distribution has a variable with nonnegligible influence.

More precisely, let $f\colon\{0,1\}^n\to[0,1]$ be a multilinear polynomial of degree at most $d$. Its variance under the uniform distribution is
\[ 
    \Var[f] 
    = 
    \E_{X} \sqbrac{\brac{f(X)-\E_{Y}[f(Y)]}^2}
    \,,
\]
and the influence of the $i$\textsuperscript{th} coordinate of the input is
\[ 
    \Inf_i[f] 
    = 
    \E_{X}\sqbrac{\brac{\vphantom{\E_{Y}}f(X)-f(X^{\oplus i})}^2} 
    \,, 
\]
where $X$ and $Y$ are independent uniformly random inputs and
$X^{\oplus i}$ is obtained from $X$ by flipping its
$i$\textsuperscript{th} bit.
The Aaronson--Ambainis conjecture asserts that
\[
    \max_{i}
    \Inf_i[f] 
    \geq
    \poly
    \brac{
        \frac{
            \Var[f]
        }{
            d
        }
    }
    ,
\]
for some universal polynomial $\poly(\cdot)$ independent of the
dimension $n$.

Aaronson and Ambainis proved that their conjecture implies the folklore conjecture stated above.
Thus, the question of classically
simulating quantum query algorithms under the uniform distribution is reduced to a question about bounded low-degree polynomials over the Boolean hypercube.
Despite recent progress~\cite{DFKO07, Mon12, OZ16, LZ23, BSdW22, Bha25, BDST26, LM26}, the Aaronson--Ambainis conjecture remains a central open problem.

\paragraph{Our Results}
We investigate whether the uniform input distribution is essential to the folklore conjecture. In particular, how much can the input distribution be biased or otherwise structured before polynomial-overhead classical simulation becomes impossible and exponential quantum advantage becomes possible? 

We study two natural families of distributions:
\begin{enumerate}
    \item The $p$-biased product distribution $\mu_p$, under which the input bits are independent and each is equal to $1$ with probability $p$.
    \item The slice distribution $\nu_p$, which is uniform over the set of $n$-bit strings of Hamming weight $\lfloor pn\rfloor$.
\end{enumerate}
These families represent two different departures from uniformly random inputs: $\mu_p$ introduces a bias while preserving independence, whereas $\nu_p$ imposes a global Hamming-weight constraint and hence introduces correlations among the input bits.

Our main result is that the folklore conjecture is insensitive to these changes in the input distribution. We formulate its natural analogue under each distribution, with the approximation error measured on average over that distribution, and we prove that all of these analogues are equivalent.

\begin{theorem}[informal]
For every fixed $p,q\in(0,1)$, quantum query algorithms under the input distribution $\mu_p$ can be classically simulated with polynomial query overhead if and only if the same holds under the input distribution $\nu_q$.

Equivalently, for every fixed $p\in(0,1)$, the folklore conjecture under the uniform distribution holds if and only if its analogue holds under $\mu_p$, and, similarly, the folklore conjecture under the uniform distribution holds if and only if its analogue holds under $\nu_p$.
\end{theorem}

Thus, neither introducing a constant bias into each input bit nor conditioning the input on having a fixed Hamming weight changes whether quantum query algorithms admit polynomial-overhead classical simulations. 

In particular, if superpolynomial quantum query advantages for decision problems are impossible under the uniform distribution, then they are also impossible under all of these distributions.

We also establish the corresponding equivalence for the Aaronson--Ambainis conjecture. We formulate natural variants of the conjecture under biased product distributions and Hamming slices, with variance and influence defined relative to the corresponding distribution, and show that they are all equivalent.

\begin{theorem}[informal]
For every fixed $p,q\in(0,1)$, the Aaronson--Ambainis conjecture under the input distribution $\mu_p$ is equivalent to the Aaronson--Ambainis conjecture under the input distribution $\nu_q$, where variance and influence are measured with respect to the corresponding distribution.

Equivalently, for every fixed $p\in(0,1)$, the original Aaronson--Ambainis conjecture under the uniform distribution holds if and only if its analogue holds under $\mu_p$, and similarly under~$\nu_p$.
\end{theorem}

Together, these results show that both the folklore conjecture and the Aaronson--Ambainis conjecture are robust across these natural families of input distributions.

Finally, our results imply that if the Aaronson--Ambainis conjecture is true under one of the input distributions $\mu_p$ or $\nu_p$, for {\em some} $p\in(0,1)$, then 
quantum query algorithms under each of the input distributions $\mu_p$ and $\nu_p$ can be classically simulated with polynomial query overhead, for {\em every} $p\in(0,1)$.
This follows immediately by the equivalence of all analogues of the Aaronson--Ambainis conjecture and by the equivalence of all analogues of the folklore conjecture, but a priori, it's not even clear how to prove that the Aaronson--Ambainis conjecture under  $\nu_p$ implies the folklore conjecture under the same $\nu_p$.

\subsection{Organization}

In Section~\ref{sec:prelims}, we give some definitions and preliminaries.
In Section~\ref{sec:simul_equiv}, we prove an equivalence between the distributional variants of the folklore conjecture.
In Section~\ref{sec:aa}, we prove an equivalence between the distributional  variants of the Aaronson--Ambainis conjecture, and discuss its implications to the folklore conjecture.



    

    



\section{Preliminaries}\label{sec:prelims}

Let $\N = \set{1,2,3,\dots}$ denote the set of natural numbers.
For $n\in \N$, let $[n]$ denote the set $\set{1,2,\dots,n}$.


\subsection{The Distribution Families We Consider}

A family of distributions $\blambda$ is a tuple $\blambda = (\lambda_n)_{n\in \N}$, where $\lambda_n$ is a distribution over $\set{0,1}^n$.
In this work, we shall be interested in two families of distributions:

\begin{definition}
	For $n\in \N,\ p\in (0,1)$, let $\mu_{n,p}$ be the $n$-fold product distribution on $\set{0,1}^n$, with each coordinate distributed as $\textnormal{Bernoulli}(p)$.

	For $n\in \N,\ p\in (0,1)$, let $\nu_{n,p}$ be the uniform distribution over the slice of Hamming-weight $\lfloor pn\rfloor$, given by $\set{x\in \set{0,1}^n: \sum_{i=1}^n x_i= \lfloor pn\rfloor}$.
	
	Define the families \[\pbias = (\mu_{n,p})_{n\in \N},\quad \pslice = (\nu_{n,p})_{n\in \N}.\]
\end{definition}


\subsection{Quantum and Classical Query Algorithms}

The reader is referred to~\cite[Section 1.1]{Ham25} for an introduction to classical and quantum query complexity.

For a quantum query algorithm $Q$, we use $Q(x)$ to denote its acceptance probability on input $x$.
We allow classical query algorithms (or decision trees) to output real values in $[0,1]$.
Throughout, query algorithms will be non-uniform and computationally unbounded except for the number of queries they make.\footnote{Most reductions we present can be made uniform.}

We prove some lemmas that shall be useful later:

\begin{lemma}[Hybrid Argument~\cite{BBBV97}]\label{lemma:hyb}
Let $\gamma\in [0,1]$, and let $(X,Y)$ be a coupling of two random variables on $\set{0,1}^n$ such that for every $i\in [n]$, $\Pr[Y_i \not= X_i \mid X] \leq \gamma$ almost surely.\footnote{Formally, this means that for every $i\in [n]$, and every $x\in \set{0,1}^n$ with $\Pr[X=x]>0$, it holds that $\Pr[Y_i \not= X_i \mid X] \leq \gamma$.\label{footnote:as}}
Then, 

\begin{enumerate}
	\item For every quantum query algorithm $Q:\set{0,1}^n\to [0,1]$ making at most $T$ queries, \[ \E\sqbrac{\abs{Q(X)-Q(Y)}} \leq 4T\sqrt{\gamma}.\]
	\item For every deterministic classical query algorithm $D:\set{0,1}^n\to [0,1]$ making at most $T$ queries, \[ \E\sqbrac{\abs{D(X)-D(Y)}} \leq T\gamma.\]
\end{enumerate}
\end{lemma}
\begin{proof}
\begin{enumerate}
	\item Without loss of generality, we may assume that the quantum algorithm $Q$ performs no intermediate measurements, and at the end outputs the result of measuring the first qubit. For any $x\in \set{0,1}^n$, let $q_{t,i}(x)$ be the query weight of index $i\in [n]$ immediately before the $t$\textsuperscript{th} query: this is the probability that measuring that query register outputs $i$; these satisfy $\sum_{i\in [n]}q_{t,i}(x)=1$.
	By the hybrid argument~\cite[Theorem 10]{Ham25}, followed by the fact that the difference between measurement outcome probabilities for two states is at most twice the Euclidean distance between the states, we have that for any $y\in \set{0,1}^n$, 
	\[ \abs{Q(x)-Q(y)} \leq 4\cdot \sum_{t\in [T]} \sqrt{\sum_{i\in [n]:x_i\not=y_i}q_{t,i}(x)}.\]
	Taking expectation over $X,Y$, and using Jensen's inequality, we get
	\begin{align*}
		\E_{X,Y}\sqbrac{\abs{Q(X)-Q(Y)}} &\leq \E_{X,Y}\sqbrac{4\cdot \sum_{t\in [T]} \sqrt{\sum_{i\in [n]}\ind\sqbrac{X_i\not=Y_i}\cdot q_{t,i}(X)}}
		\\&\leq 4\cdot \E_{X}\sqbrac{\sum_{t\in [T]} \sqrt{\sum_{i\in [n]}\Pr_Y\sqbrac{X_i\not=Y_i \mid X}\cdot q_{t,i}(X)}}
		\\&\leq  4\cdot \E_{X}\sqbrac{\sum_{t\in [T]} \sqrt{\gamma}} \leq 4\sqrt{\gamma}\cdot T.
	\end{align*}
	
	\item For any $x\in \set{0,1}^n$, let $S_x\subseteq [n]$ denote the coordinates queried by $D$ on the computational path on $x$. Since $|S_x|\leq T$, by the union bound, we have \[ \E_{X,Y}\sqbrac{\abs{D(X)-D(Y)}} \leq \E_X\sqbrac{\Pr_Y\sqbrac{Y_i\not= X_i \text{ for some } i\in S_X \mid X}} \leq T\gamma. \qedhere\]
\end{enumerate}
\end{proof}

\begin{lemma}[Derandomizing Classical Query Algorithms]\label{lemma:derand_dt}
Let $m, n \in \N$, and let $X = (X_1,\dots,X_n)\in \set{0,1}^n$ and $Z = (Z_{1},\dots,Z_{n}) \in (\set{0,1}^m)^n$ be random variables such that the pairs $(X_1,Z_1), (X_2,Z_2), \dots, (X_n, Z_n)$ are mutually independent.

Let $f:\set{0,1}^n\to [0,1]$ be a function, and let $D:\set{0,1}^{mn}\to [0,1]$ be a deterministic $T$-query classical query algorithm.
Then, there exists a deterministic $T$-query classical query algorithm $A:\set{0,1}^n\to [0,1]$ such that
\[ \E_X\sqbrac{\abs{A(X)-f(X)}} \leq \E_{X,Z}\sqbrac{\abs{D(Z)-f(X)}}. \]
\end{lemma}
\begin{proof}
	For every sequence $w = (w_{i,b})_{i\in [n], b\in \set{0,1}}$, with each $w_{i,b} \in \set{0,1}^m$, we define a deterministic classical query algorithm $A^w :\set{0,1}^n\to  [0,1]$ as follows:
	On input $x\in \set{0,1}^n$, simulate (implicitly) the query algorithm $D$ on input $(w_{1,x_1}, w_{2, x_2}, \dots, w_{n,x_n})$; that is, whenever $D$ asks for a query in the $i$\textsuperscript{th} block $w_{i,x_i}$ for the first time, $A^w$ queries $x_i$, and having learned $x_i$, it answers that query and any later query to the same block using the hardwired string $w_{i,x_i}$.
	Note that $A^w$ makes at most $T$ queries since $D$ makes at most $T$ queries.
	
	Next, define a random variable $W = (W_{i,b})_{i\in [n], b\in \set{0,1}}$ (independent of $X$ and $Z$) on such sequences as follows: for each $i\in [n], b\in \set{0,1}$, independently, sample $W_{i,b}$ from the distribution of $Z_i$ conditioned on $X_i=b$; if the event $X_i=b$ occurs with probability 0, define $W_{i,b}$ arbitrarily.
	Then, we have
	\[\E_{X,W}\sqbrac{\abs{A^W(X)-f(X)}} = \E_{X,W}\sqbrac{\abs{D(W_{1,X_1},\dots,W_{n, X_n})-f(X)}}  = \E_{X,Z}\sqbrac{\abs{D(Z)-f(X)}}.\]
	We used that the distribution of $(W_{1,X_1},\dots,W_{n, X_n})$ conditioned on $X$ is the same distribution as that of $Z$ conditioned on $X$.
	Finally, by the probabilistic method, there exists a choice $W=w$ for which $A^w$ satisfies the lemma statement.
\end{proof}


\subsection{Degree, Variance and Influence}

In this section, we define the degree, variance, and influence of real-valued functions over the Boolean hypercube.
We note that our definitions consider polynomials over the entire hypercube and not only the support of the distribution; this is especially natural for quantum query algorithms, whose acceptance probabilities are defined on every point of the hypercube.

\begin{definition}[Degree]
	We say that a function $f:\set{0,1}^n\to \R$ has degree (at most) $d$ if it can be represented as a multilinear polynomial of degree at most $d$, i.e., \[f(x) = \sum_{S\subseteq [n], |S|\leq d} a_S x^S,\] for every $x\in \set{0,1}^n$, where $x^S:=\prod_{i\in S}x_i$.

	We call a function $f$ a bounded degree-$d$ function/polynomial if it has degree at most $d$, and if $f(x)\in [0,1]$ for each input $x$.
\end{definition}

\begin{definition}[Variance and Influence]
	Let $\lambda_n$ be a distribution on $\set{0,1}^n$ and let $f:\set{0,1}^n\to \R$ be a function.
	The variance of $f$ is defined as
	\[ \Var\up{\lambda_n}[f] = \E_{X\sim \lambda_n}\sqbrac{\brac{f(X)-\E_{\lambda_n}[f]}^2} = \E_{X\sim \lambda_n}\sqbrac{f(X)^2}- \brac{\E_{X\sim \lambda_n}\sqbrac{f(X)}}^2.\]
	Note that $0\leq \Var\up{\lambda_n}[f] \leq 1/4$ for functions $f:\set{0,1}^n\to [0,1]$.
	
	For every coordinate $i\in [n]$, the influence of the $i$\textsuperscript{th} coordinate of $f$ is defined as
	\[ \Inf_i\up{\lambda_n}[f] = \E_{X\sim \lambda_n}\sqbrac{\brac{f(X)-f(X^{\oplus i})}^2} = \E_{X\sim \lambda_n}\sqbrac{\brac{\partial_if(X)}^2}, \]
	where $X^{\oplus i}$ is used to denote the input $X$ with the $i$\textsuperscript{th} bit flipped, and $\partial_i f$ is the $i$\textsuperscript{th} partial derivative of (the multilinear polynomial of) $f$.
	The total influence of $f$ is defined as \[\Inf\up{\lambda_n}[f] = \sum_{i=1}^n \Inf_i\up{\lambda_n}[f].\]
\end{definition}

Note that for the uniform distribution $\mu_{n,1/2}$, the above matches the usual notion of $L^2$-influence~\cite{Don14} up to a normalization constant.
Also, observe that for slice distributions $\nu_{n,p}$, the definition of influence depends on the values of the function outside the slice.
This is natural in this work, since quantum query algorithms are well-defined on all points of the hypercube.
In the following lemma, we compare this definition of influence with the more intrinsic and well-studied definition of \emph{swap-influence} over the slice (see~\cite{Fil16} for reference).
The lemma shows that a lower bound on the max-swap-influence already implies a lower bound on the above-defined max-influence, and hence the above definition leads to a weaker conjecture later on in this work (see Definition~\ref{defn:distributional_aa}).

\begin{lemma}\label{lemma:sw_inf_bound}
	Let $n\in \N,\ p\in (0,1)$, let $k = \lfloor pn \rfloor$, and let $\Omega = \set{x\in \set{0,1}^n: \sum_{i=1}^n x_i=k}$ be the slice of Hamming-weight $k$; recall that $\nu_{n,p}$ is the uniform distribution over $\Omega$.
	
	For each $x\in \set{0,1}^n$, and $i,j\in [n]$, let $x\up{i,j}$ denote the input $x$ with coordinates $i$ and $j$ swapped.
	For a function $f:\set{0,1}^n\to \R$, and $i,j\in [n]$, define the \emph{swap-influence} of $(i,j)$ on $f$ as:
	\[ \textnormal{swInf}_{i,j}\up{\nu_{n,p}}[f] = \E_{X\sim \nu_{n,p}}\sqbrac{\brac{f(X)-f(X\up{i,j})}^2}. \]
	Then, for every $i,j\in [n]$, it holds that
	\[ \textnormal{swInf}_{i,j}\up{\nu_{n,p}}[f] \leq  2\cdot \brac{\Inf_i\up{\nu_{n,p}}[f]+ \Inf_j\up{\nu_{n,p}}[f]} \leq 4\cdot \max_{i\in [n]}\Inf_i\up{\nu_{n,p}}[f].\]
\end{lemma}
\begin{remark}
	We note that the definition of swap-influence only relies on the values of $f$ on the slice $\Omega$.
\end{remark}
\begin{proof}	
Note that for every $x\in \set{0,1}^n$, the inputs $x$ and $x\up{i,j}$ have the same Hamming-weight; in particular, if $X\sim \nu_{n,p}$, then $X\up{i,j}\sim \nu_{n,p}$.
Moreover, for each $i,j\in [n]$,
\begin{align*}
	\E_{X\sim \nu_{n,p}}\sqbrac{\brac{f(X)-f(X\up{i,j})}^2} &\leq \E_{X\sim \nu_{n,p}}\sqbrac{\brac{f(X)-f(X^{\oplus i})+ f\brac{(X\up{i,j})^{\oplus j}} - f(X\up{i,j})}^2}
	\\&\leq 2\cdot \E_{X\sim \nu_{n,p}}\sqbrac{\brac{f(X)-f(X^{\oplus i})}^2 + \brac{f\brac{(X\up{i,j})^{\oplus j}} - f(X\up{i,j})}^2}
	\\&= 2\cdot \brac{\Inf_i\up{\nu_{n,p}}[f]+ \Inf_j\up{\nu_{n,p}}[f]} \leq 4\cdot \max_{i\in [n]}\Inf_i\up{\nu_{n,p}}[f].
\end{align*}
The first inequality holds since if $X_i=X_j$, then the left-hand side is zero; if $X_i\not=X_j$, then $X\up{i,j} = (X^{\oplus i})^{\oplus j}$, and $f(X)-f(X\up{i,j}) = f(X) - f(X^{\oplus i}) + f\brac{(X\up{i,j})^{\oplus j}} - f(X\up{i,j})$.
\end{proof}


\subsection{Some Bounds for Low-degree Polynomials}

We state the well-known Markov brothers' inequality (see~\cite{Sha04} for reference):

\begin{lemma}[Markov brothers' inequality]
Let $g:\R\to \R$ be a polynomial of degree $d$. Then,
\[ \sup_{x\in [-1,1]} \abs{g'(x)} \leq d^2 \cdot \sup_{x\in [-1,1]} \abs{g(x)}.\]
\end{lemma}
\begin{corollary}\label{corr:markov_brothers}
Let $g:\R\to \R$ be a polynomial of degree $d$. Then,
\[ \sup_{x\in [0,1]} \abs{g'(x)} \leq 2d^2 \cdot \sup_{x\in [0,1]} \abs{g(x)}.\]
\end{corollary}
\begin{proof}
	This follows by applying the Markov brothers' inequality to $h(x) = g\brac{\frac{1+x}{2}}$.
\end{proof}

We show a higher-dimensional consequence of the Markov brothers' inequality:

\begin{lemma}\label{lemma:lowdeg_flip}
	Let $f:\set{0,1}^n\to [0,1]$ be a multilinear degree-$d$ polynomial.
	Then, for every $x\in \set{0,1}^n$, and every $A\subseteq [n]$, 
	\[ \abs{\sum_{i\in A} \brac{f(x)-f(x^{\oplus i})}} \leq 2d^2, \]
	where $x^{\oplus i}$ denotes $x$ with the $i$\textsuperscript{th} coordinate flipped.
\end{lemma}
\begin{proof}
	Fix $x\in \set{0,1}^n$ and $A\subseteq [n]$.
	For $t\in [0,1]$, let $X(t) \in \set{0,1}^n$ be obtained by randomly flipping each coordinate of $x$ in the set $A$ with probability $t$, and leaving coordinates outside $A$ unchanged.
	Define a polynomial $g:\R\to \R$ as 
	\[ g(t) = \E[f(X(t))].\]
	This is a univariate polynomial of degree at most $d$, and satisfies $\sup_{t\in [0,1]}\abs{g(t)}\leq 1$. Moreover, $g'(0) = \sum_{i\in A} \brac{f(x^{\oplus i})-f(x)}$.
	The result now follows by Markov brothers' inequality (Corollary~\ref{corr:markov_brothers}).
\end{proof}

We show an upper bound on the total influence of a bounded low-degree function:

\begin{lemma}[A Total Influence Upper Bound]\label{lemma:tot_inf_bound}
	Let $\lambda_n$ be a distribution on $\set{0,1}^n$ and let $f:\set{0,1}^n\to [0,1]$ be a bounded degree-$d$ polynomial.
	Then,
	\[ \Inf\up{\lambda_n}[f] \leq 4d^2 .\]
\end{lemma}
\begin{proof}
	It suffices to prove that for every $x\in \set{0,1}^n$,  $\sum_{i=1}^n\brac{f(x)-f(x^{\oplus i})}^2 \leq 4d^2$.
	
	Fix some $x\in \set{0,1}^n$ and let $a_i = f(x^{\oplus i})-f(x)$ for each $i\in [n]$.
	By Lemma~\ref{lemma:lowdeg_flip}, for any set of coordinates $A\subseteq [n]$, it holds that
	$ \abs{\sum_{i\in A} a_i} \leq 2d^2. $
	Using that $\abs{a_i}\leq 1$ for each $i$, we have
	\[ \sum_{i=1}^n a_i^2 \leq \sum_{i=1}^n \abs{a_i} = \Big|\sum_{i:a_i>0} a_i\Big| + \Big|\sum_{i:a_i<0} a_i\Big| \leq 2d^2+2d^2 = 4d^2.  \qedhere\]
\end{proof}

We also prove a (weak) lower bound on the maximum coordinate influence for our distributions of interest:

\begin{lemma}[A Max Influence Lower Bound]\label{lemma:maxinf_lb}
	Let $p\in (0,1)$, let $\lambda_n$ be either $\mu_{n,p}$ or $\nu_{n,p}$, and let $f:\set{0,1}^n\to \R$ be a function.
	Then,
	\[ \max_{i\in [n]}\Inf_i\up{\lambda_n}[f] \geq \frac{\Var\up{\lambda_n}[f]}{2n^2} . \]
\end{lemma}
\begin{proof}
	Let $v = \Var\up{\lambda_n}[f]$, and let $M = \max_{i\in [n]}\Inf_i\up{\lambda_n}[f]$.
	Consider the following cases:
	\begin{enumerate}
		\item Suppose $\lambda_n = \mu_{n,p}$.
		Let $X, Y\sim \mu_{n,p}$ be independent.
		For each $i\in \set{0,1,\dots,n}$, define $Z\up{i} = (X_1,\dots,X_i, Y_{i+1}, \dots,Y_n)$; observe that each $Z\up{i}\sim \mu_{n,p}$.
		By the Cauchy-Schwarz inequality, we have
		\begin{align*}
			2\cdot \Var\up{\mu_{n,p}}[f] &=  \E\sqbrac{(f(X)-f(Y))^2}
			\\&= \E\sqbrac{\brac{\sum_{i\in [n]}\brac{f(Z\up{i})-f(Z\up{i-1})}}^2}
			\\&\leq n\cdot \sum_{i\in [n]}\E\sqbrac{\brac{f(Z\up{i})-f(Z\up{i-1})}^2}
			\\&= n\cdot  \sum_{i\in [n]} \Pr[X_i\not=Y_i]\cdot \Inf_i\up{\mu_{n,p}}[f]
			\\&\leq n\cdot n\cdot 1\cdot M \leq n^2M,
		\end{align*}
		as desired.
		We used that if $X_i=Y_i$, then $Z\up{i} = Z\up{i-1}$; if $X_i\not= Y_i$, then $Z\up{i} = (Z\up{i-1})^{\oplus i}$ and $\abs{f(Z\up{i})-f(Z\up{i-1})} = \abs{\partial_if(Z\up{i-1})}$ does not depend on the inputs $X_i,Y_i$ in coordinate $i$.
		
		\item Suppose $\lambda_n = \nu_{n,p}$. For each $x\in \set{0,1}^n$, and $i,j\in [n]$, we use $x\up{i,j}$ to denote the input $x$ with coordinates $i$ and $j$ swapped.
		Note that $x$ and $x\up{i,j}$ have the same Hamming-weight; in particular, if $X\sim \nu_{n,p}$, then $X\up{i,j}\sim \nu_{n,p}$.
		By Lemma~\ref{lemma:sw_inf_bound}, 
		this satisfies that for each $i,j\in [n]$,
		\[ \E_{X\sim \nu_{n,p}}\sqbrac{\brac{f(X)-f(X\up{i,j})}^2}\leq 2\cdot \brac{\Inf_i\up{\nu_{n,p}}[f]+ \Inf_j\up{\nu_{n,p}}[f]} \leq 4M.\]
		
		Now, let $X\sim \nu_{n,p}$, let $\pi:[n]\to[n]$ be a permutation chosen uniformly at random, and let $Y = (X_{\pi(1)}, X_{\pi(2)}, \dots, X_{\pi(n)})$; then $X, Y \sim \nu_{n,p}$ are independent.
		We may write $\pi$ as a series of transpositions/swaps $(i_1,j_1), (i_2,j_2),\dots,(i_\ell,j_\ell)$, with $\ell\leq n$.
		Defining $Z\up{0}=X$, and for each $t\in \set{1,2,\dots,\ell}$, $Z\up{t} = (Z\up{t-1})\up{i_t,j_t}$, we have that $Z\up{\ell}=Y$.
		Observe that conditioned on $\pi$, each $Z\up{t}\sim \nu_{n,p}$. 
		Hence, by the Cauchy-Schwarz inequality, we have
		\begin{align*}
			\Var\up{\nu_{n,p}}[f] &= \frac{1}{2}\cdot \E_{X,\pi}\sqbrac{(f(X)-f(Y))^2}
			\\&=  \frac{1}{2}\cdot \E_{X,\pi}\sqbrac{\brac{\sum_{t\in [\ell]}\brac{f(Z\up{t})-f(Z\up{t-1})}}^2}
			\\&\leq \frac{1}{2}\cdot \E_{\pi}\sqbrac{\ell\cdot \sum_{t\in [\ell]}\E_{X}\sqbrac{\brac{f(Z\up{t})-f(Z\up{t-1})}^2}}
			\\&\leq \frac{1}{2}\cdot \E_{\pi}\sqbrac{\ell\cdot \ell\cdot 4M} \leq 2n^2M.\qedhere
		\end{align*}
	\end{enumerate}
\end{proof}
\begin{remark}
	We note that the quadratic dependence on $n$ can be improved to a linear dependence for both distributions, and we only work with the above bounds to keep the proof self-contained.
	\begin{enumerate}
		\item For the biased hypercube, elementary Fourier analysis~\cite[Sections 8.3-8.4]{Don14} implies the Poincar\'e inequality $ \Var\up{\mu_{n,p}}[f] \leq p(1-p) \cdot \sum_{i=1}^n \Inf_i\up{\mu_{n,p}}[f].$
		\item For the slice, the following Poincar\'e inequality is known~\cite[Lemma 3.14]{FM19}:\footnote{In the work~\cite{FM19}, the inequality is stated for harmonic polynomials. Since any polynomial over the slice has a harmonic extension, and since both sides of the inequality depend only on values of $f$ on the slice, the inequality applies here as well.}
	\[ \Var\up{\nu_{n,p}}[f] \leq \frac{1}{2n}\sum_{i<j} \E_{X\sim \nu_{n,p}}\sqbrac{\brac{f(X)-f(X\up{i,j})}^2} .\]
	By Lemma~\ref{lemma:sw_inf_bound}, the right-hand side is bounded by $n\cdot \max_{i\in [n]}\Inf_i\up{\nu_{n,p}}[f]$.
	\end{enumerate}
\end{remark}


\section{Classical Simulation of Quantum Algorithms}\label{sec:simul_equiv}

\begin{definition}[Classical Simulation of Quantum Query Algorithms]\label{defn:qcsim}
	Let $\blambda = (\lambda_n)_{n\in \N}$ be a family of distributions, where $\lambda_n$ is a distribution over $\set{0,1}^n$, and let $S:\N\times (0,1) \to \N$ be a function.
	We say that quantum query algorithms can be $(\blambda,S)$-simulated by classical query algorithms if the following holds:
	
	Let $n\in \N,\ \epsilon \in (0,1),\ T\in \N$, and let $Q:\set{0,1}^n \to [0,1]$ be any quantum query algorithm making at most $T$ queries.
	Then, there exists a deterministic classical query algorithm $D:\set{0,1}^n \to [0,1]$ making at most $S(T, \epsilon)$ queries and such that
	\[ \E_{X\sim \lambda_n}\sqbrac{\abs{D(X)-Q(X)}} \leq \epsilon. \]
	Note  the function $S$ does not depend on the input length $n$ directly.
\end{definition}

By Markov's inequality, the above also implies that on \emph{most inputs}, the classical algorithm's answer is (additively) close to the quantum algorithm's answer. 

\begin{definition}[Simulation with Polynomial Overhead]
	We say that quantum query algorithms over a family of distributions $\blambda = (\lambda_n)_{n\in \N}$ can be \emph{polynomially classically simulated} if they can be $(\blambda, S)$-simulated for some function $S$ satisfying \[S(T, \epsilon) \leq \poly(T, 1/\epsilon),\] where the polynomial is independent of $n$.	
\end{definition}

Our main result is the following:

\begin{theorem}\label{thm:equiv_prod_slice}
	For any $p,q\in (0,1)$, the following are equivalent:
	\begin{enumerate}
		\item Quantum query algorithms over $\pbias$ can be polynomially classically simulated.
		\item Quantum query algorithms over $\qslice$ can be polynomially classically simulated.
	\end{enumerate}	
	In particular, the families $\pbias, \qbias, \pslice, \qslice$ are all equivalent with respect to polynomial classical simulation.
\end{theorem}
\begin{remark}
	We note that the resulting simulation overheads in the above theorem depend polynomially on $1/\min(p,1-p,q,1-q)$.
\end{remark}

\begin{remark}\label{remark:general_dist}
	More generally, our proofs also show an equivalence of the above families of distributions with families $\blambda = (\lambda_n)_{n\in \N}$ of the following form: Let $a\in (0,1/2]$ be a constant, and for each $n\in \N$, let $\lambda_n$ be either $\mu_{n,p_n}$ or $\nu_{n,p_n}$, with $p_n\in [a,1-a]$. 

    More generally still, the same conclusion holds when, for every $n\in \N$, the distribution $\lambda_n$ is either a slice distribution $\nu_{n,p_n}$, with $p_n\in [a,1-a]$, or a product distribution with bounded coordinate-dependent biases, i.e., $\lambda_n = \bigotimes_{i=1}^n \textnormal{Bernoulli}(p_n^{(i)})$, with each $p_n^{(i)} \in [a,1-a]$.
\end{remark}

The remainder of this section is devoted to the proof of the above theorem.


\subsection{Equivalence of Simulation over Product Measures with Different Biases}\label{sec:gadget_comp}

In this section, we prove that the families $\pbias$ and $\qbias$ are equivalent with respect to classical simulation of quantum query algorithms with polynomial overhead.

The main idea behind the proof is the following: A single $q$-biased bit can be (approximately) simulated by a function $g$ taking as input a small number of $p$-biased bits (Lemma~\ref{lemma:inner_fn}); this gadget is motivated by a similar one used by Keller~\cite{Kel12}.
Then, given any function $f:\set{0,1}^n\to [0,1]$, to understand its ``behavior'' over the $\qbias$, it suffices to understand the behavior of $f\circ g^n$ over $\pbias$.
Finally, to convert the approximation to an exact statement, we may use continuity or the hybrid argument (Lemma~\ref{lemma:hyb}). 

\begin{proposition}\label{prop:equiv_prod}
	Let $p,q\in (0,1)$.
	Suppose that quantum query algorithms can be $(\pbias, S)$-simulated by classical query algorithms.
	Then, quantum query algorithms can be $(\qbias, S')$-simulated by classical query algorithms, with 
	\[ S'(T, \epsilon) \leq S\brac{2T\cdot \left\lceil\frac{2}{a}\log_2\brac{\frac{8T}{\epsilon \sqrt{a}}}\right\rceil, \frac{\epsilon}{2}} ,\]
	where $a = \min(p,1-p,q,1-q)$.
	
	In particular, if quantum query algorithms over $\pbias$ can be polynomially classically simulated, then quantum query algorithms over $\qbias$ can be polynomially classically simulated.
\end{proposition}

The remainder of this section is devoted to the proof of the above proposition.
For this, we fix some $p,q\in (0,1)$, and let $a = \min(p,1-p,q,1-q) \in (0,1/2]$.
Also, suppose that quantum query algorithms can be $(\pbias, S)$-simulated by classical query algorithms.

Let $m\in \N$ be an integer parameter.
First, we show that the $q$-biased distribution can be well-approximated by a Boolean function over the $p$-biased distribution:

\begin{lemma}\label{lemma:inner_fn}
	There exists a function $g:\set{0,1}^m \to \set{0,1}$ such that 
	\[ q \leq \Pr_{z\sim \mu_{m,p}}\sqbrac{g(z)=1} \leq q+2^{-am} .\]
\end{lemma}
\begin{proof}
	List all elements $z\in \set{0,1}^m$ in an arbitrary order.
	Define $g$ to be the indicator function of the smallest prefix of elements with measure at least $q$.
	The prefix measure can overshoot $q$ by an amount that is at most the measure of the final added element, which is at most $(1-a)^m \leq 2^{-am}$ under the distribution $\mu_{m,p}$.
\end{proof}

Let $n\in \N$ and let $Q:\set{0,1}^n\to [0,1]$ be a quantum query algorithm making $T$ queries.
Let $g:\set{0,1}^m \to \set{0,1}$ be a function as in the above lemma, and let \[r := \Pr_{z\sim \mu_{m,p}}\sqbrac{g(z)=1} \in \sqbrac{q, q+2^{-am}}.\]
With this, we show the following quantum-to-classical simulation:

\begin{lemma}\label{lemma:bias_classical_alg}
	For every $\epsilon\in (0,1)$, there exists a deterministic classical query algorithm $A$ making at most $S(2mT, \epsilon)$ queries, and such that
	\[ \E_{Y\sim \mu_{n,r}}\sqbrac{\abs{A(Y)-Q(Y)}} \leq \epsilon. \]
\end{lemma}

\begin{proof}

Define $G: \set{0,1}^{mn}\to \set{0,1}^n$ to be the $n$-fold product of the function $g$, and $\tilde{Q}:\set{0,1}^{nm}\to [0,1]$ to be the corresponding lifted quantum query algorithm, given by:
\[ G(z) = \brac{g(z_{1}),g(z_{2}),\dots,g(z_{n})}, \quad  \tilde{Q}(z) = Q(G(z)), \]
where $z = (z_1,z_2,\dots,z_n)$ with each $z_{i}\in \set{0,1}^m$.
Observe that $\tilde{Q}$ makes at most $2mT$ queries to its input, since one query to the $i$\textsuperscript{th} bit of $G(z)$, given by $g(z_i)$, can be simulated by $2m$ queries to $z$: this is done by querying all the bits of $z_{i}$, computing $g(z_{i})$, and then uncomputing the queried block and all auxiliary workspace.

By the assumed simulation property over the $p$-biased distribution, we get a deterministic classical query algorithm $D:\set{0,1}^{mn}\to [0,1]$ making at most $S(2mT, \epsilon)$ queries, and such that 
\[ \E_{Z\sim \mu_{mn, p}}\sqbrac{\abs{D(Z)-Q(G(Z))}} = \E_{Z\sim \mu_{mn, p}}\sqbrac{|D(Z)-\tilde{Q}(Z)|} \leq \epsilon.\]
Observe that when $Z\sim \mu_{mn, p}$, it holds that $Y := G(Z) \sim \mu_{n,r}$.
Hence, the above implies
\[\E_{Y,Z}\sqbrac{\abs{D(Z)-Q(Y)}} \leq \epsilon.\]
Finally, since the pairs $(Y_1,Z_1), \dots, (Y_n,Z_n)$ are mutually independent, derandomization (Lemma~\ref{lemma:derand_dt}) implies the existence of a classical query algorithm $A$ as desired.
\end{proof}

We can now prove the main result of this subsection:

\begin{proof}[Proof of Proposition~\ref{prop:equiv_prod}]
	Consider any $\epsilon\in (0,1)$, and choose $m = \left\lceil\frac{2}{a}\log_2\brac{\frac{8T}{\epsilon \sqrt{a}}}\right\rceil,\ \epsilon'=\epsilon/2$.
	Let $A$ be the deterministic classical algorithm as in Lemma~\ref{lemma:bias_classical_alg}, making at most $S(2mT, \epsilon')$ queries, and such that $\E_{Y\sim \mu_{n,r}}\sqbrac{\abs{A(Y)-Q(Y)}} \leq \epsilon'$.
	
	Define a coupling $(X,Y)$ of random variables, with $X\sim \mu_{n,q}$ and $Y\sim \mu_{n,r}$ as follows:
	For each $i\in [n]$ independently, let $U_i \in [0,1]$ be a uniform random variable, and let $X_i = \ind[U_i\leq q]$ and $Y_i = \ind[U_i\leq r]$.
	This coupling satisfies that for each $i\in [n]$, almost surely\footref{footnote:as}
	\[ \Pr[X_i \not= Y_i \mid X] \leq \frac{\abs{q-r}}{\min(q,1-q)} \leq \frac{2^{-am}}{a} .\]
	
	By the hybrid argument (Lemma~\ref{lemma:hyb}) applied to $Q$, and Lemma~\ref{lemma:bias_classical_alg}, we get
	\[ \E\sqbrac{\abs{A(Y)-Q(X)}} \leq \E\sqbrac{\abs{A(Y)-Q(Y)}}  + \E\sqbrac{\abs{Q(X)-Q(Y)}} \leq \epsilon/2 + 4T \cdot \frac{2^{-am/2}}{\sqrt{a}} \leq \epsilon . \]
	Finally, applying Lemma~\ref{lemma:derand_dt} (since the pairs $(X_1,Y_1), \dots, (X_n,Y_n)$ are mutually independent), we get an $S(2mT, \epsilon')$-query deterministic classical query algorithm $B$ such that
	\[ \E_{X\sim \mu_{n,q}}\sqbrac{\abs{B(X)-Q(X)}} \leq \epsilon . \qedhere \]
\end{proof}


\subsection{Equivalence of Simulation over Biased Product Measure and the Slice}

In this section, we show that the families of distributions $\pbias$ and $\pslice$ are equivalent with respect to classical simulation of quantum query algorithms with polynomial overhead.
The main idea is to show that quantum query algorithms have similar expectations over the two distributions.
Towards this, we first show a coupling between the two distributions:

\begin{lemma}[Slice and Biased Product Coupling]\label{lemma:slice_prod_coup}
Let $p\in (0,1)$, let $a = \min(p,1-p) \in (0,1/2]$, and let $n\in \N$.
Then, there exists a coupling $(X,Y)$ of random variables on $\set{0,1}^n$ such that $X\sim \nu_{n,p}$ and $Y\sim \mu_{n,p}$, and such that for every $i\in [n]$, almost surely\footref{footnote:as}
\[ \Pr[Y_i\not=X_i \mid X] \leq \frac{4}{\sqrt{an}}. \]
\end{lemma}

\begin{proof}
	Suppose that $n \geq 16/a$, or else the lemma holds as the left-hand side is at most 1.

	Let $X \sim \nu_{n,p}$ be a uniformly random string of Hamming-weight $k= \lfloor pn \rfloor$.
	Conditioned on $X$, sample $Y$ as follows: Sample $W \sim \textnormal{Binomial}(n,p)$; if $W=k$, define $Y=X$; if $W > k$, choose $W-k$ coordinates uniformly at random from the 0-coordinates of $X$, flip them to 1, and let the resulting string be $Y$; if $W < k$, choose $k-W$ coordinates uniformly at random from the 1-coordinates of $X$, flip them to 0, and let the resulting string be $Y$.
	
	Observe that conditioned on $W=w$, the distribution of $Y$ is uniform over all strings of Hamming-weight $w$.
	Since $W \sim \textnormal{Binomial}(n,p)$, it holds that $Y\sim \mu_{n,p}$.
	
	Next, conditioned on $X$, we have the following for each $i\in [n]$: 
	
	\begin{enumerate}
		\item If $X_i=1$, then 
		\[\Pr[Y_i\not=X_i \mid X] = \E_W\sqbrac{\ind\sqbrac{W<k}\cdot \frac{k-W}{k}} \leq \frac{\E\abs{W-k}}{k} \leq \frac{\E\abs{W-k}}{\min(k, n-k)} .\]		
		\item If $X_i=0$, then \[\Pr[Y_i\not=X_i \mid X] = \E_W\sqbrac{\ind\sqbrac{W>k}\cdot \frac{W-k}{n-k}} \leq \frac{\E\abs{W-k}}{n-k}\leq \frac{\E\abs{W-k}}{\min(k, n-k)} .\]
	\end{enumerate}
	Also, we have that
	\[ \E\abs{W-k} \leq \abs{pn-k}+ \E\abs{W-pn} \leq 1 + \sqrt{\E\sqbrac{\brac{W-pn}^2}} \leq 1 + \sqrt{p(1-p)n} \leq 1+\sqrt{an}.\]
	Combining the above inequalities, we get
	\[ \Pr[Y_i\not=X_i \mid X] \leq \frac{\E\abs{W-k}}{\min(k, n-k)} \leq  \frac{\sqrt{an}+1}{an-1} \leq  \frac{2\sqrt{an}}{an/2}\leq \frac{4}{\sqrt{an}}. \qedhere\]
\end{proof}

Next, we prove the main result of this section:

\begin{proposition}\label{prop:equiv_prod_and_slice}
	Let $p\in (0,1)$, and let $a = \min(p,1-p)$.
	\begin{enumerate}
		\item Suppose that quantum query algorithms can be $(\pbias, S)$-simulated by classical query algorithms.
			Then, quantum query algorithms can be $(\pslice, S')$-simulated by classical query algorithms, with \[ S'(T, \epsilon) \leq 10^{7}\cdot a^{-1}\epsilon^{-4}\cdot \max\set{S(T,\epsilon/2)^2, T^4}.\]
			
		\item Suppose that quantum query algorithms can be $(\pslice, S)$-simulated by classical query algorithms.
			Then, quantum query algorithms can be $(\pbias, S')$-simulated by classical query algorithms, with \[ S'(T, \epsilon) \leq 10^{7}\cdot a^{-1}\epsilon^{-4}\cdot \max\set{S(T,\epsilon/2)^2, T^4}.\]
	\end{enumerate}
		
	In particular, quantum query algorithms over $\pbias$ can be polynomially classically simulated if and only if quantum query algorithms over $\pslice$ can be polynomially classically simulated.
\end{proposition}

\begin{proof}
We start with the first part.
Let $p\in (0,1)$, and let $a = \min(p,1-p) \in (0,1/2]$.
Suppose that quantum query algorithms can be $(\pbias, S)$-simulated by classical query algorithms.
Let $n\in \N,\ \epsilon\in (0,1)$, and let $Q:\set{0,1}^n\to [0,1]$ be a quantum query algorithm making at most $T$ queries.

Let $N:= \left\lceil\frac{2^{20}}{a\epsilon^4}\cdot \max\set{S(T,\epsilon/2)^2, T^4}\right\rceil$.
If $n < N$, a classical query algorithm $D$ can query all bits of the input $x\in \set{0,1}^n$, and output $Q(x)$. This makes at most $N$ queries and has zero error.

Now, suppose that $n \geq N$.
By our assumption, there exists a deterministic classical query algorithm $D$, making at most $S(T, \epsilon/2)$ queries, and such that 
\[ \E_{Y\sim \mu_{n,p}}\sqbrac{\abs{D(Y)-Q(Y)}} \leq  \epsilon/2.\]
We will show that the same classical query algorithm $D$, making at most $S(T, \epsilon/2)$ queries, works over the slice $\nu_{n,p}$.
Using Lemma~\ref{lemma:slice_prod_coup}, we have a coupling $(X,Y)$ with $X\sim\nu_{n,p}$ and $Y\sim \mu_{n,p}$, and such that for every $i\in [n]$, $\Pr[Y_i\not=X_i \mid X] \leq \frac{4}{\sqrt{an}}$.
Then, by the hybrid argument (Lemma~\ref{lemma:hyb}), we have
\begin{align*}
	\E\sqbrac{\abs{D(X)-Q(X)}} &\leq \E\sqbrac{\abs{D(Y)-Q(Y)}} + \E\sqbrac{\abs{D(X)-D(Y)}} + \E\sqbrac{\abs{Q(X)-Q(Y)}}  
	\\&\leq \frac{\epsilon}{2} + S(T, \epsilon/2)\cdot \frac{4}{\sqrt{an}} + 4T\cdot \frac{2}{(an)^{1/4}} \leq \frac{\epsilon}{2}+\frac{\epsilon}{4}+\frac{\epsilon}{4}= \epsilon.
\end{align*}

The proof of the second part of this proposition is identical: we use the same coupling, and transfer the slice simulator from $X$ to $Y$.
\end{proof}

The main theorem of this section now follows:
\begin{proof}[Proof of Theorem~\ref{thm:equiv_prod_slice}]
	This follows from Proposition~\ref{prop:equiv_prod} and Proposition~\ref{prop:equiv_prod_and_slice}.	
\end{proof}


\section{Influences in Bounded Low-Degree Polynomials}\label{sec:aa}

Aaronson and Ambainis conjectured that every bounded low-degree polynomial on the Boolean hypercube has an influential coordinate with respect to the uniform measure~\cite{AA14}.
They showed that, assuming this conjecture, quantum query algorithms over $\halfbias$ can be polynomially classically simulated.

In this section, we formulate a version of their conjecture over general distributions, and show equivalences among these conjectures for natural distributions.
Then, we show how the equivalences lead to polynomial classical simulations over the distributions.

\begin{definition}[Distributional Aaronson--Ambainis]\label{defn:distributional_aa}
	Let $\blambda = (\lambda_n)_{n\in \N}$ be a family of distributions, where $\lambda_n$ is a distribution over $\set{0,1}^n$, and let $S:(0,1/4]\times \N \to [0,1]$.
	We say that $\AA(\blambda, S)$ holds if the following is true:
	For every $n,d\in \N$, $v\in (0,1/4]$, and every bounded degree-$d$ polynomial $f:\set{0,1}^n\to [0,1]$ with $\Var\up{\lambda_n}[f]\geq v$, there exists a coordinate $i\in [n]$ with
	\[ \Inf_i\up{\lambda_n}[f] \geq S(v, d). \]
	Note  the function $S$ does not depend on the input length $n$ directly.
	
	We say that $\AA(\blambda)$ holds if $\AA(\blambda, S)$ holds with some polynomially large function $S$, i.e., there exist constants $c, C>0$ such that $S(v, d) \geq c\cdot \brac{\frac{v}{d}}^C$ for all $v\in (0,1/4]$ and $d\in \N$.
\end{definition}

Thus the original Aaronson--Ambainis conjecture is:

\begin{conjecture}[{\cite[Conjecture 1.7]{AA14}}]
	$\AA(\halfbias)$ holds.
\end{conjecture}
 
The main result of this section is the following:

\begin{theorem}\label{thm:aa_equiv_prod_slice}
	For any $p,q\in (0,1)$, we have $\AA\brac{\pbias} \iff \AA\brac{\qslice}$.
	In particular, the families $\pbias, \qbias, \pslice, \qslice$ are all equivalent with respect to the $\AA$ criterion.
\end{theorem}
\begin{remark}
	We note that the overheads in the transferred polynomial lower bounds in the above theorem depend polynomially on $\min(p,1-p,q,1-q)$.
\end{remark}

\begin{remark}\label{remark:p_dep_on_n}
    As in Remark~\ref{remark:general_dist}, the conclusion of the above theorem holds with more general families $\blambda = (\lambda_n)_{n\in \N}$, where for each $n\in \N$, the distribution $\lambda_n$ is either a slice distribution, or a product distribution whose coordinate biases are uniformly bounded away from 0 and 1.
\end{remark}

We prove the above theorem in Section~\ref{sec:aa1} and Section~\ref{sec:aa2}; the arguments parallel the proofs in Section~\ref{sec:simul_equiv}.
Then, in Section~\ref{sec:simulation}, we show applications to polynomial classical simulation.


\subsection{Aaronson--Ambainis Equivalence over Product Measures with Different Biases}\label{sec:aa1}

In this section, we prove that the families $\pbias$ and $\qbias$ are equivalent with respect to the $\AA$ criterion.
The proof reuses the gadget-composition argument in Section~\ref{sec:gadget_comp}.

We first prove that variance and influence vary continuously with the bias of the distribution:

\begin{lemma}\label{lemma:varinf_cont}
	Let $f:\set{0,1}^n\to [0,1]$ be a degree-$d$ polynomial, and let $q,r\in [0,1]$.
	Then, for every $i\in [n]$, it holds that 
	\[ \abs{\Var\up{\mu_{n,q}}[f]-\Var\up{\mu_{n,r}}[f]} \leq 8d^2\cdot \abs{q-r}, \quad
	\abs{\Inf_i\up{\mu_{n,q}}[f]-\Inf_i\up{\mu_{n,r}}[f]} \leq 8d^2\cdot \abs{q-r} .\]
\end{lemma}
\begin{proof}
	Fix some $i\in [n]$, and define
	\[ V(t) = \Var\up{\mu_{n,t}}[f], \qquad I_i(t) = \Inf_i\up{\mu_{n,t}}[f], \qquad t\in [0,1].\]
	Since the expectation of a multilinear monomial $x^S$ under $\mu_{n,t}$ is $t^{|S|}$, both $ V(t)$ and $I_i(t)$ are univariate polynomials in $t$, of degree at most $2d$.
	Moreover, for each $t\in [0,1]$, we have $0\leq V(t) \leq 1/4$ and $0\leq I_i(t) \leq 1$.
	The result now follows by the Markov brothers' inequality (Corollary~\ref{corr:markov_brothers}) and the mean-value theorem.
\end{proof}

Next, we prove a reduction from $\pbias$ to $\qbias$:

\begin{proposition}\label{prop:equiv_prod_aa}
	Let $p,q\in (0,1)$ and let $a = \min(p,1-p,q,1-q)$.
	Suppose that $\AA(\pbias, S)$ holds.
	Then, for every $D\in \N$, $\AA(\qbias, S_D)$ holds with 
	\[ S_D(v, d) \geq S\brac{v/2,\ d \cdot m_D(v,d)} - (v/d)^D ,\]
	where $m_D(v,d) =  \lceil \frac{3D}{a}\log_2(\frac{d}{v})\rceil$.
\end{proposition}

\begin{proof}
Fix some $p,q\in (0,1)$, and let $a = \min(p,1-p,q,1-q) \in (0,1/2]$.
Suppose that $\AA(\pbias, S)$ holds.

Let $n\in \N$ and let $f:\set{0,1}^n\to [0,1]$ be a degree-$d$ polynomial with $\Var\up{\mu_{n,q}}[f] \geq v$.
Let $D\in \N$, let $m = m_D(v,d) =  \lceil \frac{3D}{a}\log_2(\frac{d}{v})\rceil \in \N$, and let $g:\set{0,1}^m \to \set{0,1}$ be a function as in Lemma~\ref{lemma:inner_fn}, with 
	\[r := \Pr_{z\sim \mu_{m,p}}\sqbrac{g(z)=1} \in \sqbrac{q, q+2^{-am}}.\]

Let $\delta = 8d^2 2^{-am}$.
By the definition of $m$, and using $v\leq 1/4$, we have
\[ \delta \leq 8d^2 \cdot (v/d)^{3D} \leq 8d^2 \cdot (1/(4d))^2\cdot (v/d)^D \leq 1/2\cdot (v/d)^D \leq \min\set{v/2, (v/d)^D}. \]
By Lemma~\ref{lemma:varinf_cont}, it holds that for every $i\in [n]$,
\begin{equation}\label{eq:varinf_cont}
	\abs{\Var\up{\mu_{n,q}}[f]-\Var\up{\mu_{n,r}}[f]} \leq \delta, \quad
	\abs{\Inf_i\up{\mu_{n,q}}[f]-\Inf_i\up{\mu_{n,r}}[f]} \leq \delta .
\end{equation}

Define $G: \set{0,1}^{mn}\to \set{0,1}^n$ to be the $n$-fold product of the function $g$, and $\tilde{f}:\set{0,1}^{nm}\to [0,1]$ to be the corresponding function composed with $f$, given by:
\[ G(z) = \brac{g(z_{1}),g(z_{2}),\dots,g(z_{n})}, \quad  \tilde{f}(z) = f(G(z)), \]
where $z = (z_1,z_2,\dots,z_n)$ with each $z_{i}\in \set{0,1}^m$.
The function $\tilde{f}$ satisfies:

\begin{enumerate}
	\item The degree of $\tilde{f}$ is at most $\deg(f)\cdot \deg(g) \leq dm$.
	\item $\Var\up{\mu_{nm,p}}[\tilde{f}] = \Var\up{\mu_{n,r}}[f]$. By Equation~\ref{eq:varinf_cont}, this is at least $v - \delta \geq v/2$.
	\item For any coordinate $(i,j)\in [n]\times [m]$, it holds by the chain rule that 
	\[ \partial_{ij} \tilde{f}(z) =  \partial_i f\brac{G(z)} \cdot \partial_j g(z_i). \]
	The first factor on the right-hand side depends only on blocks $z_\ell$ for $\ell\not=i$, and the second factor depends only on the block $z_i$. Hence, squaring and taking expectation under $\mu_{nm, p}$, we get
	\[ \Inf_{(i,j)}\up{\mu_{nm,p}}[\tilde{f}] = \Inf_i\up{\mu_{n,r}}[f] \cdot \Inf_j\up{\mu_{m,p}}[g] \leq  \Inf_i\up{\mu_{n,r}}[f] . \]
	Combining with Equation~\ref{eq:varinf_cont}, we get
	\[ \Inf_{(i,j)}\up{\mu_{nm,p}}[\tilde{f}] \leq \Inf_i\up{\mu_{n,q}}[f] +  \delta \leq  \Inf_i\up{\mu_{n,q}}[f] + (v/d)^D  .\]
\end{enumerate}

Now, by our assumption $\AA(\pbias, S)$, there exists $(i,j)\in [n]\times [m]$ such that \[ \Inf_{(i,j)}\up{\mu_{nm,p}}[\tilde{f}] \geq S(v/2,\ dm).\]
Using the above points, we get 
\[ \Inf_i\up{\mu_{n,q}}[f] \geq S(v/2, dm) - (v/d)^D . \qedhere \]
\end{proof}

Finally, by an appropriate choice of parameters, we prove the main result of this section:

\begin{corollary}\label{corr:aa_equiv_prod}
	For every $p,q\in (0,1)$, $\AA\brac{\pbias} \iff \AA\brac{\qbias}$.
\end{corollary}
\begin{proof}
	Fix some $p,q\in (0,1)$, and let $a = \min(p,1-p,q,1-q) \in (0,1/2]$.
	Suppose that $\AA(\pbias, S)$ holds with $S(v, d) \geq c\cdot\brac{\frac{v}{d}}^C$ for some constants $c,C$.
	
	By Proposition~\ref{prop:equiv_prod_aa}, for every $D\in \N$, $\AA(\qbias, S_D)$ holds with 
	\[ S_D(v, d) \geq S\brac{v/2,\ d \cdot m_D(v,d)} - (v/d)^D ,\]
	where $m_D(v,d) =  \lceil \frac{3D}{a}\log_2(\frac{d}{v})\rceil$.
	
	Using the notation $x= v/d\in (0, 1/4]$, we have $m_D(v,d) \leq \frac{6D}{a}\log_2(1/x) \leq \frac{6D}{ax}$. Hence,
	\begin{align*}
		S_D(v,d) &\geq c \cdot \brac{\frac{v}{2d\cdot m_D(v,d)}}^C - \brac{\frac{v}{d}}^D
		\\&\geq c\cdot \brac{\frac{ax^2}{12D}}^C - x^D
		\\& \geq \brac{\frac{1}{2}\cdot \frac{c\cdot a^C}{(12D)^C}}\cdot x^{2C}.
	\end{align*}
	The last inequality holds for a large enough integer $D$, satisfying $\brac{\frac{1}{4}}^{D-2C}\cdot D^C \leq \brac{\frac{1}{2}\cdot \frac{c\cdot a^C}{12^C}}$.
	Note that the final influence bound depends polynomially on $a$.
	
	The reverse implication follows by interchanging $p$ and $q$.
\end{proof}


\subsection{Aaronson--Ambainis Equivalence over Biased Product Measure and the Slice}\label{sec:aa2}

In this section, we prove that the families $\pbias$ and $\pslice$ are equivalent with respect to the $\AA$ criterion.

As a first step, we show that bounded low-degree polynomials are close in expectation over the two distributions.
For this, we analyze how the expectation of a polynomial differs over two slices of the hypercube.
\begin{lemma}\label{lemma:lowdeg_two_slice}
	Let $f:\set{0,1}^n\to [0,1]$ be a bounded degree-$d$ polynomial. For each $t\in \set{0,1,2,\dots,n}$, let $F(t)$ denote the expectation of $f$ over the Hamming-weight $t$ slice.
	Then, for each $t\in [n]$, it holds that $\abs{F(t)-F(t-1)} \leq 4d^2/n.$
	
	In particular, for each $t,u \in  \set{0,1,2,\dots,n}$, it holds that 
	\[\abs{F(t)-F(u)} \leq \frac{4d^2}{n}\cdot \abs{t-u}.\]
\end{lemma}
\begin{proof}
	Let $t\in [n]$.
	Observe that if we pick a uniformly random input $X\in \set{0,1}^n$ of Hamming-weight $t$, and flip a random coordinate $i$ with $X_i=1$, we obtain a uniformly random input of Hamming-weight $t-1$.
	Hence, by Lemma~\ref{lemma:lowdeg_flip}, we have
	\begin{align*}
		\abs{F(t-1)-F(t)} &= \abs{\E_{X:|X|=t}\sqbrac{\frac{1}{t} \sum_{i\in [n]:X_i=1} f(X^{\oplus i})} - \E_{X:|X|=t}\sqbrac{f(X)}}
		\\&\leq \frac{1}{t}\cdot \max_{x\in \set{0,1}^n} \abs{\sum_{i\in [n]:x_i=1} \brac{f(x^{\oplus i})-f(x)}} \leq \frac{2d^2}{t}.
	\end{align*}
	A similar argument, where we first pick a random input of Hamming-weight $t-1$, and then flip a 0-coordinate, shows that $\abs{F(t-1)-F(t)} \leq \frac{2d^2}{n-t+1}$.
	Combining the two inequalities, we get
	\[ \abs{F(t)-F(t-1)} \leq \frac{2d^2}{\max(n-t+1, t)} \leq \frac{4d^2}{n}. \qedhere\]
\end{proof}

\begin{lemma}\label{lemma:lowdeg_slice_vs_biased}
	Let $p\in (0,1)$ and $f:\set{0,1}^n\to [0,1]$ be a degree-$d$ polynomial. Then,
	\[ \abs{\E_{\mu_{n,p}}[f] - \E_{\nu_{n,p}}[f]} \leq \frac{8d^2}{\sqrt{n}}. \]
\end{lemma}
\begin{proof}
	For each $t=0,1,2,\dots,n$, let $F(t)$ denote the expectation of $f$ over the Hamming-weight $t$ slice.
	Then, for $k=\lfloor pn \rfloor$, and $W\sim\mathrm{Binomial}(n,p)$, we have by Lemma~\ref{lemma:lowdeg_two_slice} that
	\begin{align*}
		\abs{\E_{\mu_{n,p}}[f] - \E_{\nu_{n,p}}[f]} &=   \abs{\E_W[F(W)]-F(k)} 
		\\&\leq \E_W\abs{F(W)-F(k)}
		\\&\leq \frac{4d^2}{n}\cdot \E\abs{W-k}.
	\end{align*}	
	The result follows by observing that
	\[ \E\abs{W-k} \leq \abs{pn-k}+ \E\abs{W-pn} \leq 1 + \sqrt{\E\sqbrac{\brac{W-pn}^2}} \leq 1 + \sqrt{p(1-p)n} \leq 2\sqrt{n}. \qedhere\]
\end{proof}

We apply the above lemma to conclude that variance and coordinate influences are close over the two distributions.

\begin{corollary}\label{corr:lowdeg_varinf_slice_biased}
	Let $p\in (0,1)$ and let $f:\set{0,1}^n\to [0,1]$ be a degree-$d$ polynomial. Then, for each $i\in [n]$,
	\[ \abs{\Var\up{\mu_{n,p}}[f] - \Var\up{\nu_{n,p}}[f] } \leq \frac{48d^2}{\sqrt{n}}, \qquad \abs{\Inf_i\up{\mu_{n,p}}[f] - \Inf_i\up{\nu_{n,p}}[f]} \leq \frac{48d^2}{\sqrt{n}}. \]
\end{corollary}
\begin{proof}
	Applying Lemma~\ref{lemma:lowdeg_slice_vs_biased} to both $f$ and $f^2$, we have
	\begin{align*}
		\abs{\Var\up{\mu_{n,p}}[f] - \Var\up{\nu_{n,p}}[f] } &\leq \abs{\E_{\mu_{n,p}}[f^2] - \E_{\nu_{n,p}}[f^2] }+ \abs{\E_{\mu_{n,p}}[f]^2 - \E_{\nu_{n,p}}[f]^2 }
		\\&\leq \abs{\E_{\mu_{n,p}}[f^2] - \E_{\nu_{n,p}}[f^2] }+ 2\abs{\E_{\mu_{n,p}}[f] - \E_{\nu_{n,p}}[f] }
		\\&\leq \frac{8(2d)^2}{\sqrt{n}} + 2\cdot \frac{8d^2}{\sqrt{n}} \leq \frac{48d^2}{\sqrt{n}}.
	\end{align*}
	The influence estimate is obtained by applying Lemma~\ref{lemma:lowdeg_slice_vs_biased} to the bounded degree-$2d$ polynomial $(f(x)-f(x^{\oplus i}))^2$.
\end{proof}

With this, we prove the main result of this section:

\begin{proposition}\label{prop:aa_equiv_prod_and_slice}
	Let $p\in (0,1)$.
	\begin{enumerate}
		\item Suppose that $\AA(\pbias, S)$ holds. Then, $\AA(\pslice, S')$ holds, with \[ S'(v,d) \geq 10^{-9}d^{-8}\cdot \min\set{v^5, v\cdot S(v/2, d)^4}.\]
		\item Suppose that $\AA(\pslice, S)$ holds. Then, $\AA(\pbias, S')$ holds, with \[ S'(v,d) \geq 10^{-9}d^{-8}\cdot \min\set{v^5, v\cdot S(v/2, d)^4}.\]	
	\end{enumerate}
	In particular, $\AA\brac{\pbias} \iff \AA\brac{\pslice}$.
\end{proposition}
\begin{proof}
We start with the first part.
Let $p\in (0,1)$, and suppose that $\AA(\pbias, S)$ holds.
Let $n\in \N$ and let $f:\set{0,1}^n\to [0,1]$ be a degree-$d$ polynomial with $\Var\up{\nu_{n,p}}[f]\geq v$.
Assume that $S(v/2,d)>0$, or else there is nothing to prove.

Let $N = \lceil 10^4d^4/\min\set{v, S(v/2, d)}^2 \rceil$. If $n<N$, by Lemma~\ref{lemma:maxinf_lb} we know that there exists a coordinate $i\in [n]$ with influence at least $ v/(2n^2)\geq 10^{-9}d^{-8}\cdot \min\set{v^5, v\cdot S(v/2, d)^4}$.

Now, suppose that $n\geq N$. Then, $\delta := 48d^2/\sqrt{n} \leq 0.5\cdot \min\set{v,\ S(v/2, d)}$, and by Corollary~\ref{corr:lowdeg_varinf_slice_biased}, we have that $\Var\up{\mu_{n,p}}[f] \geq v - \delta \geq v/2$.
Hence, by our assumption $\AA(\pbias, S)$, there exists a coordinate $i\in [n]$ such that \[\Inf_i\up{\mu_{n,p}}[f] \geq S(v/2,d).\]
Then, by Corollary~\ref{corr:lowdeg_varinf_slice_biased}, we have
\[ \Inf_i\up{\nu_{n,p}}[f] \geq S(v/2,d) - \delta \geq 0.5\cdot S(v/2,d) \geq 10^{-9}d^{-8}\cdot \min\set{v^5, v\cdot S(v/2, d)^4}. \]

The proof of the second part of this proposition is identical.
\end{proof}

The main theorem of this section now follows:

\begin{proof}[Proof of Theorem~\ref{thm:aa_equiv_prod_slice}]
	This follows from Corollary~\ref{corr:aa_equiv_prod} and Proposition~\ref{prop:aa_equiv_prod_and_slice}.
\end{proof}


\subsection{Polynomial Classical Simulation under Aaronson--Ambainis}\label{sec:simulation}

As noted above, Aaronson and Ambainis showed polynomial simulation under their conjecture on influences in bounded low-degree polynomials.
Formally, they showed:
\begin{theorem}[{\cite[Theorem 3.3]{AA14}}]\label{thm:aa_paper_simulation}
	$\AA(\halfbias)$ implies that quantum query algorithms over $\halfbias$ can be polynomially classically simulated.
\end{theorem}
\begin{proof}[Proof Sketch]
	Suppose $\AA(\halfbias)$ holds and let $Q:\set{0,1}^n\to [0,1]$ be a quantum query algorithm making at most $T$ queries.
	Then, $Q$ is a multilinear polynomial of degree at most $2T$~\cite[Lemma 4.2]{BBCMdW01}.
	The classical query algorithm $D$ does the following: Let $\eta = \poly(\epsilon)$ be a parameter, and
	\begin{enumerate}
		\item If $Q$ has low variance, that is $\Var\up{\mu_{n,1/2}}[Q]\leq \eta$, then the classical query algorithm $D$ outputs the constant $\E_{X\sim \mu_{n,1/2}}[Q(X)]$.
			In this case, 
			\[\E\abs{Q(X)-D(X)} \leq \sqrt{\E\abs{Q(X)-D(X)}^2} = \sqrt{\Var\up{\mu_{n,1/2}}[Q]}\leq \sqrt{\eta}.\]
		\item Else, if $\Var\up{\mu_{n,1/2}}[Q]> \eta$, by the assumption $\AA(\halfbias)$, there exists a coordinate $i\in [n]$ with $\Inf_i\up{\mu_{n,1/2}}[Q] \geq \poly(\eta/T)$.
	In this case, $D$ queries coordinate $i$, and then repeats the above steps with the polynomial $Q$ restricted to the queried value of $X_i$.
	\end{enumerate}		
	Using the total influence as a potential function, along with Lemma~\ref{lemma:tot_inf_bound}, it can be shown that with high probability the process terminates after $\poly(T,1/\epsilon)$ steps.
	Then, truncating the decision tree at a suitably larger polynomial depth yields a tree with a worst case polynomial query bound, and error at most $\epsilon$.
\end{proof}

Using the equivalences proved in Theorem~\ref{thm:aa_equiv_prod_slice} and Theorem~\ref{thm:equiv_prod_slice}, similar results for the slice and biased-product distributions directly follow from the above result over $\halfbias$:

\begin{theorem}
	Let $\AA$ be used to denote either of the (equivalent) statements $\AA(\qbias)$ or $\AA(\qslice)$ for some $q\in (0,1)$.

	Then, $\AA$ implies that for every $p\in (0,1)$, quantum query algorithms over $\pbias$ and $\pslice$ can be polynomially classically simulated.
\end{theorem}

\begin{remark}
	The proof sketch of Theorem~\ref{thm:aa_paper_simulation} also works for the family $\pbias$ for any $p\in (0,1)$, assuming $\AA(\pbias)$ holds.
	This is because conditioning on the queried coordinates leaves the other coordinates independent and $p$-biased.
	
	A direct recursion on the slice family $\pslice$, while possible, is much more intricate if we only assume $\AA(\pslice)$.
	This is because conditioning a $k$-slice on the queried coordinates produces slices with different parameters.
	One can get around this by using Theorem~\ref{thm:aa_equiv_prod_slice} (see Remark~\ref{remark:p_dep_on_n}) if all the densities stay in some range $[a,1-a]$, and this can be maintained by truncating the recursion appropriately.
	Essentially, our equivalence theorems lead to a much cleaner proof.	
\end{remark}

\section*{AI Disclosure}
We used Google Gemini and ChatGPT to assist with the proofs in this paper.
Google Gemini suggested the core idea underlying the inner-gadget construction in Lemma~\ref{lemma:inner_fn}, and ChatGPT assisted in developing the remaining proofs and refining the exposition.
The tools materially affected all the technical results in this paper.
The authors verified the correctness and originality of all content including the references.

\bibliographystyle{alpha}
\bibliography{main.bib}

@article {AA14,
    AUTHOR = {Aaronson, Scott and Ambainis, Andris},
     TITLE = {The need for structure in quantum speedups},
   JOURNAL = {Theory Comput.},
  FJOURNAL = {Theory of Computing. An Open Access Journal},
    VOLUME = {10},
      YEAR = {2014},
     PAGES = {133--166}
}

@article {Fil16,
    AUTHOR = {Filmus, Yuval},
     TITLE = {An orthogonal basis for functions over a slice of the {B}oolean hypercube},
   JOURNAL = {Electron. J. Combin.},
  FJOURNAL = {Electronic Journal of Combinatorics},
    VOLUME = {23},
      YEAR = {2016},
    NUMBER = {1},
     PAGES = {Paper 1.23, 27}
}

@misc {Ham25,
    AUTHOR = {Hamoudi, Yassine},
     TITLE = {A Brief Introduction to Quantum Query Complexity},
      YEAR = {2025},
       NOTE = {Available at \url{https://arxiv.org/pdf/2508.08852}.}
}

@book {Don14,
    AUTHOR = {O'Donnell, Ryan},
     TITLE = {Analysis of {B}oolean functions},
 PUBLISHER = {Cambridge University Press, New York},
      YEAR = {2014}
}

@incollection {Sha04,
    AUTHOR = {Shadrin, Aleksei},
     TITLE = {Twelve proofs of the {M}arkov inequality},
 BOOKTITLE = {Approximation theory: a volume dedicated to {B}orislav {B}ojanov},
     PAGES = {233--298},
 PUBLISHER = {Prof. M. Drinov Acad. Publ. House, Sofia},
      YEAR = {2004}
}

@article {FM19,
    AUTHOR = {Filmus, Yuval and Mossel, Elchanan},
     TITLE = {Harmonicity and invariance on slices of the {B}oolean cube},
   JOURNAL = {Probab. Theory Related Fields},
  FJOURNAL = {Probability Theory and Related Fields},
    VOLUME = {175},
      YEAR = {2019},
    NUMBER = {3-4},
     PAGES = {721--782}, 
      NOTE = {(also in CCC 2016)}
}

@article {BBCMdW01,
    AUTHOR = {Beals, Robert and Buhrman, Harry and Cleve, Richard and Mosca, Michele and de Wolf, Ronald},
     TITLE = {Quantum lower bounds by polynomials},
   JOURNAL = {J. ACM},
  FJOURNAL = {Journal of the ACM},
    VOLUME = {48},
      YEAR = {2001},
    NUMBER = {4},
     PAGES = {778--797}
}

@article {DFKO07,
    AUTHOR = {Dinur, Irit and Friedgut, Ehud and Kindler, Guy and O'Donnell, Ryan},
     TITLE = {On the {F}ourier tails of bounded functions over the discrete cube},
   JOURNAL = {Israel J. Math.},
  FJOURNAL = {Israel Journal of Mathematics},
    VOLUME = {160},
      YEAR = {2007},
     PAGES = {389--412},
      NOTE = {(also in STOC 2006)}
}

@article {Sim97,
    AUTHOR = {Simon, Daniel R.},
     TITLE = {On the power of quantum computation},
   JOURNAL = {SIAM J. Comput.},
  FJOURNAL = {SIAM Journal on Computing},
    VOLUME = {26},
      YEAR = {1997},
    NUMBER = {5},
     PAGES = {1474--1483},
     NOTE = {(also in FOCS 1994)}
}

@article {YZ24,
    AUTHOR = {Yamakawa, Takashi and Zhandry, Mark},
     TITLE = {Verifiable quantum advantage without structure},
   JOURNAL = {J. ACM},
  FJOURNAL = {Journal of the ACM},
    VOLUME = {71},
      YEAR = {2024},
    NUMBER = {3},
     PAGES = {Art. 20, 50},
     NOTE = {(also in FOCS 2022)}
}

@inproceedings {LZ23,
    AUTHOR = {Lovett, Shachar and Zhang, Jiapeng},
     TITLE = {Fractional certificates for bounded functions},
 BOOKTITLE = {ITCS},
    VOLUME = {251},
     PAGES = {Art. No. 84, 13},
      YEAR = {2023},
}

@article {Mon12,
    AUTHOR = {Montanaro, Ashley},
     TITLE = {Some applications of hypercontractive inequalities in quantum information theory},
   JOURNAL = {J. Math. Phys.},
  FJOURNAL = {Journal of Mathematical Physics},
    VOLUME = {53},
      YEAR = {2012},
    NUMBER = {12},
     PAGES = {122206, 15}
}

@inproceedings {OZ16,
    AUTHOR = {O'Donnell, Ryan and Zhao, Yu},
     TITLE = {Polynomial bounds for decoupling, with applications},
 BOOKTITLE = {CCC},
    VOLUME = {50},
     PAGES = {Art. No. 24, 18},
      YEAR = {2016},
}

@inproceedings {BSdW22,
    AUTHOR = {Bansal, Nikhil and Sinha, Makrand and de Wolf, Ronald},
     TITLE = {Influence in completely bounded block-multilinear forms and classical simulation of quantum algorithms},
 BOOKTITLE = {CCC},
    VOLUME = {234},
     PAGES = {Art. No. 28, 21},
      YEAR = {2022}
}

@misc {LM26,
    AUTHOR = {Liu, Qipeng and Mutreja, Saachi},
     TITLE = {Parallel Quantum Advantage with Limited Adaptivity Requires Structure},
      YEAR = {2026},
       NOTE = {Available at \url{https://arxiv.org/abs/2608.20297v1}.}
}

@misc {BDST26,
    AUTHOR = {Blanc, Guy and Docter, Jordan and Strassle, Carmen and Tan, Li-Yang},
     TITLE = {Quantum Speedups Require Structure or Depth},
      YEAR = {2026},
       NOTE = {FOCS 2026 (Accepted). Available at \url{https://arxiv.org/abs/2608.19158}.}
}

@inproceedings {Bha25,
    AUTHOR = {Bhattacharya, Sreejata Kishor},
     TITLE = {Random restrictions of bounded low degree polynomials are juntas},
 BOOKTITLE = {ITCS},
    VOLUME = {325},
     PAGES = {Art. No. 17, 21},
      YEAR = {2025}
}

@article {BBBV97,
    AUTHOR = {Bennett, Charles H. and Bernstein, Ethan and Brassard, Gilles and Vazirani, Umesh},
     TITLE = {Strengths and weaknesses of quantum computing},
   JOURNAL = {SIAM J. Comput.},
  FJOURNAL = {SIAM Journal on Computing},
    VOLUME = {26},
      YEAR = {1997},
    NUMBER = {5},
     PAGES = {1510--1523}
}

@article {Kel12,
    AUTHOR = {Keller, Nathan},
     TITLE = {A simple reduction from a biased measure on the discrete cube
              to the uniform measure},
   JOURNAL = {European J. Combin.},
  FJOURNAL = {European Journal of Combinatorics},
    VOLUME = {33},
      YEAR = {2012},
    NUMBER = {8},
     PAGES = {1943--1957}
}

\end{document}